\documentclass[12pt]{amsart}
\usepackage[mathscr]{eucal}
\usepackage{amsmath,amsfonts,amssymb,amsthm}
\usepackage[matrix,arrow]{xy}
\usepackage{times}
\usepackage{float,subfigure}
\usepackage{hyperref,curves}
\usepackage{graphicx}
\usepackage{epstopdf}
\usepackage{multibox}
\usepackage{dcolumn}

\advance\textheight\topskip \textwidth 36pc \oddsidemargin 20pt
\newtheorem{prop}{Proposition}[section]
\newtheorem{lemma}[prop]{Lemma}
\newtheorem{definition}[prop]{Definition}
\newtheorem{theorem}[prop]{Theorem}
\theoremstyle{remark}

\newtheorem{rem}{Remark}[section]

\renewcommand{\simeq}{\cong}

\begin{document}

\def\mytitle{All Stable Characteristic Classes \\of Homological Vector Fields}

\pagestyle{myheadings}
\markboth{\textsc{\small E. Mosman and A. Sharapov}}{%
  \textsc{\small All Stable Characteristic Classes }}
\addtolength{\headsep}{4pt}

\begin{centering}
.

  \vspace{4cm}

  \textbf{\Large{\mytitle}}

  \vspace{1.5cm}

  {\large Elena Mosman and Alexey Sharapov}

\vspace{.5cm}

\begin{minipage}{.9\textwidth}\small \it \begin{center}
   Department of Quantum Field Theory, Tomsk State University,\\
   Lenin ave. 36, 634050 Tomsk, Russia
   \end{center}
\end{minipage}

\end{centering}

\vspace{1cm}

\begin{center}
  \begin{minipage}{.9\textwidth}
    \textsc{Abstract}. An odd vector field $Q$ on a supermanifold
    $M$ is called homological, if $Q^2=0$. The operator of Lie
    derivative $L_Q$ makes the algebra of smooth tensor fields on $M$ into a
    differential tensor algebra. In this paper, we give a complete
    classification of certain invariants of homological vector fields
    called characteristic classes. These take values in the
    cohomology of the operator $L_Q$ and are represented by
    $Q$-invariant tensors made up of the homological vector field and
    a symmetric connection on $M$ by means  of the algebraic tensor operations
    and covariant differentiation.
  \end{minipage}
\end{center}


\vfill

\noindent \mbox{}
\raisebox{-3\baselineskip}{%
  \parbox{\textwidth}{\mbox{}\hrulefill\\[-4pt]}}
{\scriptsize We are thankful to anonymous referees for useful
remarks. The work was partially supported by the RFBR grant
09-02-00723-a, by the grant from Russian Federation President
Programme of Support for Leading Scientific Schools no 871.2008.02,
and also by Russian Federal Agency of Education under the State
Contracts no P1337, no P2596 and no P22. EM appreciates financial
support from Dynasty Foundation. }

\thispagestyle{empty}
\newpage

\section{Introduction}

The unique existence theorem for solutions of ordinary
differential equations ensures integrability of smooth
one-dimensional distributions on differentiable manifolds. A new
phenomenon arises in the category of smooth supermanifolds: Due to
an extra sign factor in the definition of the supercommutator of
vector fields, the classical Frobenius criterion of integrability
\begin{equation}\label{integrability}
[Q,Q]=0
\end{equation}
is not fulfilled automatically for odd vector fields $Q$. Rather it
becomes a nontrivial condition to satisfy, $[Q,Q]=2Q^2=0$.  An odd
vector field $Q$ that squares to zero is called a
\textit{homological vector field}. The homological vector fields
were first introduced by Shander \cite{Shand} in his study of
differential equations on supermanifolds. The local normal forms of
homological vector fields were then considered by Schwarz \cite{S}
and Vaintrob \cite{Vaintrob2}. In the former  paper it was proposed
to refer to supermanifolds with homological  vector fields as
\textit{$Q$-manifolds}. The $Q$-manifolds play a prominent  role
both in physics and mathematics. In theoretical physics, the
homological vector fields appear usually as classical BRST
differentials on the ghost-extended configuration/phase spaces of
gauge theories \cite{S}, \cite{HT}. On the other hand, various
mathematical concepts can be reformulated and studied in terms of
$Q$-manifolds. An incomplete list of examples includes de Rham and
Koszul complexes, $L_{\infty}$-algebras \cite{Kon}, \cite{LSt},
rational homotopy types \cite{Sull}, Lie algebroids
\cite{Vaintrob1}, and $n$-algebroids \cite{Roy}, \cite{Sev}. The
advantage of the ``homological'' point of view over traditional ones
is its geometric clarity and flexibility. Notice that the
$Q$-manifolds form a category, whose morphisms are just
diffeomorphisms of supermanifolds that relate homological vector
fields. Having translated some class of mathematical objects in the
language of $Q$-manifolds we get a natural definition of morphisms
for the objects of interest, which may be hard to see or formulate
in a classical (i.e., non-homological) approach. A typical example
is the category of Lie algebroids where the homological approach
offers a concise and elegant formulations for such important notions
as a Lie algebroid homomorphism and an adjoint module
\cite{Vaintrob1}. From this perspective it is desirable to have a
structure theory of $Q$-manifolds, which would capture both the
local and \textit{global} properties of homological vector fields.

In this paper, we study the global  invariants of $Q$-manifold
called \textit{characteristic classes} \cite{LS}, \cite{LMS1},
\cite{LMS2}. The idea behind the construction of such invariants is
as follows. Given a $Q$-manifold $(M,Q)$, we denote by
$\mathcal{T}(M)$ the algebra of smooth tensor fields on $M$ of
arbitrary types $(n,m)$. The operator of Lie derivative $\delta=L_Q$
makes the algebra $\mathcal{T}(M)=\bigoplus \mathcal{T}^{n,m}(M)$
into a differential tensor algebra. Let
$H(M,Q)=\mathrm{Ker}\delta/\mathrm{Im}\delta$ denote the group of
$\delta$-cohomology. Since $\delta$ respects the tensor operations
-- tensor product, contraction and permutation of tensor indices --
the space $H(M,Q)$ inherits the structure of tensor algebra. The
algebra $H(M,Q)$ is thus a natural invariant of the $Q$-manifold
$M$\footnote{Of course, if we want to treat $H$ as a functor from
the category of $Q$-manifolds to the category of tensor algebras,
then we should restrict ourselves to the subalgebra of covariant
tensor fields on $M$.}. Unfortunately, this invariant is hard to
compute even in a topologically trivial situation. This is due to
possible local singularities of the homological vector field. The
way out is to consider a special differential subalgebra
$\mathcal{A}\subset \mathcal{T}(M)$ called the \textit{algebra of
concomitants}. Given a symmetric affine connection $\nabla$ on $M$
with curvature $R$, by a \textit{concomitant} associated to the
triple $(M,Q,\nabla)$ we understand a tensor field on $M$ which is
made up of the homological vector field $Q$, the curvature tensor
$R$, and their covariant derivatives by means of the tensor
operations. According to the classical reduction theorem \cite{Sch}
the concomitants exhaust all the natural tensor fields associated to
$Q$ and $\nabla$. Also the  set of all concomitants is invariant
under the action of $\delta$. We say that a $\delta$-closed
concomitant $\mathcal{C}\in \mathcal{A}$ is a \textit{universal
cocycle} if the closedness condition $\delta \mathcal{C}=0$ follows
from the integrability condition (\ref{integrability}) regardless of
any specificity of $Q$, $\nabla$ and $M$. In other words, the
universal cocycles are universal $Q$-invariant tensor polynomials in
$\nabla^nQ$ and $\nabla^mR$ that can be attributed to \textit{any }
$Q$-manifold with connection. (This, a little  bit vague, definition
of ``universality'' can be made precise using the notion of a
\textit{graph complex} associated to the differential algebra of
concomitants $\mathcal{A}$, see Section \ref{GrCom}.) The
\textit{stable} characteristic classes of $Q$-manifolds are now
defined to be the elements of $H(M,Q)$ that are represented by the
universal cocycles\footnote{We use the adjective ``stable'' to
emphasize that there may exist other characteristic classes which
are specific to $Q$-manifolds of any particular dimension, see
Remark \ref{rem} below.}. A remarkable fact \cite{LMS2} is that the
$\delta$-cohomology classes of universal cocycles do not depend on
the choice of the symmetric connection and hence they are invariants
of a $Q$-manifold as such. Furthermore, the nontrivial universal
cocycles admit a fairly explicit description in contrast to the
group $H(M,Q)$.

The algebra of concomitants is naturally graded,
$\mathcal{A}=\bigoplus \mathcal{A}_{n,m}^k$; here the subscripts
$(n,m)$ refer to the tensor type of concomitants, while the
superscript $k$ is the degree of homogeneity in $Q$ and its
derivatives:
\begin{equation*}
\mathcal{A}^k\ni \mathcal{C}[Q]\quad \Leftrightarrow\quad
t^k\mathcal{C}[Q]=\mathcal{C}[tQ]\qquad \forall t\in \mathbb{R}\,.
\end{equation*}
Since $\delta: \mathcal{A}_{n,m}^k\rightarrow
\mathcal{A}_{n,m}^{k+1}$, we have a direct sum of finite-dimensional
complexes. The stable characteristic classes form a subgroup
$H_{\mathrm{st}}(\mathcal{A})$ in the $\delta$-cohomology group
$H(\mathcal{A})=\bigoplus H(\mathcal{A}^k_{n,m})$. The computation
of the groups $H(\mathcal{A})$ and $H_{\mathrm{st}}(\mathcal{A})$ is
somewhat facilitated by the fact that the algebra of concomitants
contains a differential ideal $\mathcal{R}\subset \mathcal{A}$
generated by the concomitants of the symmetric connection
$\{\nabla^m R\}$. The corresponding short exact sequence of
complexes
\begin{equation*}
\xymatrix{0
\ar[r]&{\mathcal{R}}\ar[r]^-{i}&{\mathcal{A}}\ar[r]^-{p} &
{\mathcal{A}}/{\mathcal{R}}\ar[r]&0}
\end{equation*}
gives rise to the exact triangle in cohomology
\begin{equation*}
\xymatrix{H(\mathcal{R})\ar[rr]^{i_\ast}& & H(\mathcal{A}) \ar[dl]^{p_\ast}\\
&H(\mathcal{A}/\mathcal{R})\ar[ul]^{\partial}& }
\end{equation*}
Geometrically, one can view $H(\mathcal{A}/\mathcal{R})$ as the
space of characteristic classes of flat $Q$-manifolds, i.e.,
$Q$-manifolds admitting a flat symmetric connection. The universal
cocycles of flat $Q$-manifolds are constructed from the
$(1,n)$-tensors $\nabla^n Q$, which are symmetric in lower indices.
If the connecting homomorphism $\partial$ is nonzero, not any
characteristic class can be extended from the flat to arbitrary
$Q$-manifolds and the obstruction to extendability is controlled by
the elements of $\mathrm{Im}\,\partial$. In \cite{LMS2}, the
extendable characteristic classes, or more precisely the elements of
$H(\mathcal{A})/\mathrm{Im} i_\ast$, were called \textit{intrinsic}.
The intrinsic characteristic classes survive on flat $Q$-manifolds,
therefore they are more closely related to the structure of the
homological vector field rather than the topology of $M$. In the
stable situation, both the intrinsic characteristic classes and the
characteristic classes of flat $Q$-manifolds were explicitly
computed in \cite{LMS2}. What has remained an open question is
whether there are nontrivial universal cocycles lying in
$\mathcal{R}$. In this paper, we show that the answer is negative so
that all the stable characteristic classes  are in fact intrinsic.
The proof of this fact is given in Section \ref{5} and does not
exploit the short exact sequence above; instead, we use a special
generating set for the algebra of concomitants that reduces the
problem to the computation of a certain graph cohomology. Together
with the results of \cite{LMS2} this gives the complete
classification of stable characteristic classes of $Q$-manifolds.

\subsection*{Conventions and notation}

Throughout the paper we work in the category of smooth
supermanifolds.  This allows us to omit the boring prefix
``super'' whenever possible. So the terms like  manifolds,
functions, algebras and so on will actually mean the corresponding
notions of supergeometry.

Given a manifold $M$, we denote by $\frak X(M)$ the space of
smooth vector fields on $M$. The space $\frak X(M)$ carries both
the structure of a real Lie algebra with respect to the commutator
of vector fields and the structure of a $C^\infty(M)$-module. The
endomorphisms of the module $\frak{X}(M)$ form an associative
algebra $\frak{A}(M)$ over $C^\infty(M)$. The elements of
$\frak{A}(M)$ are smooth tensor fields of type $(1,1)$. The
operation of contraction of tensor indices endows $\frak{A}(M)$
with the natural trace $\mathrm{Str}:\frak{A}(M)\rightarrow
C^\infty(M)$. We let $\Omega(M)=\bigoplus \Omega^n(M)$ denote the
algebra of exterior differential forms on $M$.

 In our study of the
differential tensor algebra of concomitants $\mathcal{A}$ we will
mostly deal with the smooth tensor fields of type $(1,n)$. Any
such tensor field is naturally identified with a
$C^\infty(M)$-linear map from $\frak{X}(M)^{\otimes n}$ to
$\frak{X}(M)$ and, in the sequel, we will freely use this
identification. We let $O_S$ denote the fully symmetric part of a
$(1,n)$-tensor $O$ so that $(O_S)_S=O_S$. Then $O_S$ defines a map
from the symmetrized tensor power $\frak{X}(M)^{\odot n}$ to
$\frak{X}(M)$. Given a $(1,n)$-tensor field $O(X_1,\ldots,X_n)$,
we denote by $(\nabla^k
O)(X_1,\ldots,X_k)(X_{k+1},\ldots,X_{k+n})$ the value of the $k$th
covariant derivative of $O$ evaluated on $k+n$ vector fields
$X_1,\ldots,X_{k+n}$. To avoid any ambiguity concerning sign
factors let us write an explicit expression in terms of local
coordinates,
\begin{equation*}
(\nabla^k O)(X_1,\ldots
,X_k)(X_{k+1},\ldots,X_{k+n})=(-1)^{\varepsilon}
X_{k+n}^{i_{k+n}}\cdots X_1^{i_1}
\nabla_{i_1}\cdots\nabla_{i_k}O^j_{i_{k+1}...i_{k+n}}\frac{\partial}{\partial
x^j}\,,
\end{equation*}
\begin{equation*}
\varepsilon=\sum_{a}\epsilon(X_a)(\epsilon_{i_{a+1}}+\cdots+\epsilon_{i_{k+n}})\,.
\end{equation*}
All covariant or partial derivatives are assumed to act from the left.

\section{The Losik-Janyska-Markl basis of concomitants}

The elementary concomitants $\{\nabla^n Q$, $\nabla^m R\}$
generating the algebra $\mathcal{A}$ are not free; rather they
satisfy an infinite number of tensor relations coming from the
integrability condition for $Q$ and the Bianchi-Ricci identities
for $\nabla$. To take into account these relations we pass on to
another generating set of concomitants, which, similar to the set
of elementary concomitants,  consists of two parts. One part was
introduced, in fact, by Losik in the context of Gelfand-Fuks
cohomology \cite{Lo}, while the other appeared in the recent paper
by Janyska and Markl \cite{JM}.

Define the curvature tensor of $\nabla$ by
\begin{equation}\label{curvature}
R(X_1,X_2,X_3) =
([\nabla_{X_3},\nabla_{X_2}]X_1-\nabla_{[X_3,X_2]}X_1)\,.
\end{equation}
Losik's part of the generating set concerns the concomitants that
involve the covariant derivatives of $Q$. Following \cite{Lo}, we
introduce a sequence of symmetric $(1,n)$-tensors $Q_n:
\frak{X}(M)^{\odot n}\rightarrow \frak{X}(M)$ of the form
\begin{equation}\label{Q-type}
Q_n(X_1,\ldots,X_n)=
(\nabla^{n}Q)_S(X_1,\ldots,X_n)-(\nabla^{n-2}R_Q)_S(X_1,\ldots,X_{n-2})(X_{n-1},X_n)\,,
\end{equation}
where the $(1,2)$-tensor  $R_Q$ is defined by
\begin{equation*}
R_Q(X_1,X_2)=(-1)^{\epsilon(X_1)+\epsilon(X_2)}R(X_1,X_2,Q)
\end{equation*}
and the second term in the r.h.s. of (\ref{Q-type}) is absent for
$n=0,1$. Besides the total symmetry in lower indices, the
concomitants (\ref{Q-type}) satisfy an infinite sequence of
algebraic relations of the form
\begin{equation}\label{QQ}
Q_n(Q,X_1,\ldots, X_{n-1})+ (\mbox{\textit{terms involving $Q_m$
with} $m<n$})=0\,.
\end{equation}
All these relations are quadratic in $Q$'s and obtained by repeated
differentiation of integrability condition (\ref{integrability}).

The special convenience of the concomitants (\ref{Q-type}) is that
they generate a tensor algebra, which is ``almost'' closed under the
action of the differential. We have
\begin{equation}\label{dQ}
\begin{array}{c}
\delta Q_0=0\,,\quad (\delta Q_1)(X)=Q_1(Q_1(X)) - \frac12
R(X,Q,Q)\,,\quad (\delta Q_2)(X_1,X_2)=0\,,\\[5mm]
(\delta Q_n)(X_1,\ldots,X_n)=O_S(X_1,\ldots, X_n)\,,\qquad n>2\,,
\end{array}
\end{equation}
where
$$
 O(X_1,\ldots,X_n)\equiv
-\sum_{k=1}^{n-2}\binom{n}{k}(-1)^{\sum_{i\leq k}\epsilon(X_i)}
Q_{k+1}({X_1,\ldots,X_k},Q_{n-k}(X_{k+1},\ldots,X_{n}))\,.
$$
The only ``bad'' concomitant is $Q_1$. Excluding $Q_1$, we get a
differential tensor subalgebra $\mathcal{Q}\subset \mathcal{A}$
generated by all $Q_n$'s with $n\neq 1$. The algebra
$(\mathcal{Q},\delta)$ enjoys an increasing filtration
$0\subset\mathcal{Q}_2\subset \mathcal{Q}_3\subset \cdots\subset
\mathcal{Q}_{\infty}=\mathcal{Q}$, where the $n$th differential
subalgebra $\mathcal{Q}_n$ is generated by the concomitants $Q_0,
Q_2,...,Q_n$. Each $\mathcal{Q}_n$ contains a differential
subalgebra $\mathcal{Q}'_n$ constituted by all the concomitants of
$\mathcal{Q}_n$ that do not involve tensor contractions of the
homological vector field $Q=Q_0$ with the other generators
$Q_2,...,Q_n$. With these definitions the main result of the paper
can be expressed by the relation
$$H_{\mathrm{st}}(\mathcal{A})=\mathcal{Q}'_2/(\mathcal{Q}'_2\cap\delta \mathcal{Q}'_3)\,.$$ In
particular, the concomitants  of the symmetric connection
$\{\nabla^n R\}$, being considered as a set of generators completing
(\ref{Q-type}) to a multiplicative basis in $\mathcal{A}$, do not
contribute to the stable cohomology at all. The proof will be given
in Section \ref{5}.

Let us now turn to Janyska-Markl's part of the multiplicative
basis in $\mathcal{A}$. It was shown in \cite{JM} that for any
symmetric connection  with curvature (\ref{curvature}) one can
associate a sequence of $(1,n)$-tensor fields $R_n$  of the form
\begin{equation}\label{R-type}
    R_n(X_1,\ldots,X_{n})=(\nabla^{n-3}R)(X_1,\ldots,X_{n-3})(X_{n-2},X_{n-1},X_n)+
    K_n\,,
\end{equation}
such that $K_n$ is made up of $R_k$ with $k<n$ and the tensors
(\ref{R-type}) enjoy the following symmetries:
\begin{description}
    \item[$(\mathrm{S}1)$] the antisymmetry in $X_{n-1}$ and $X_{n}$,
    \item[$(\mathrm{S}2)$] the cyclic symmetry in $X_{n-2}$, $X_{n-1}$ and
    $X_n$:
\begin{equation*}
    \begin{array}{l}
(-1)^{\epsilon(X_{n-1})(\epsilon(X_{n-2})+\epsilon(X_{n}))}R_n(X_1,\ldots,X_{n-3},X_{n-2}, X_{n-1},X_n)\\[3mm]
\displaystyle
+ (-1)^{\epsilon(X_{n-2})(\epsilon(X_{n})+\epsilon(X_{n-1}))}R_n(X_1,\ldots,X_{n-3},X_{n}, X_{n-2},X_{n-1}) \\[3mm]
\displaystyle +
(-1)^{\epsilon(X_{n})(\epsilon(X_{n-1})+\epsilon(X_{n-2}))}R_n(X_1,\ldots,X_{n-3},X_{n-1},
X_{n},X_{n-2})=0\,,
    \end{array}
\end{equation*}

\item[$(\mathrm{S}3)$] for $n\geq 4$, the cyclic symmetry in
$X_{n-3}$,
    $X_{n-1}$ and $X_{n}$:

\begin{equation*}
    \begin{array}{l}
(-1)^{\epsilon(X_{n-3})(\epsilon(X_{n-2})+\epsilon(X_{n}))}R_n(X_1,\ldots,X_{n-3},X_{n-2}, X_{n-1},X_n)\\[3mm]
\displaystyle
+ (-1)^{\epsilon(X_{n-1})(\epsilon(X_{n-2})+\epsilon(X_{n-3}))}R_n(X_1,\ldots,X_{n-1},X_{n-2}, X_{n},X_{n-3})\\[3mm]
\displaystyle +
(-1)^{\epsilon(X_{n})(\epsilon(X_{n-2})+\epsilon(X_{n-1}))}R_n(X_1,\ldots,X_{n},X_{n-2},
X_{n-3},X_{n-1})=0\,,
    \end{array}
\end{equation*}
 \item[$(\mathrm{S}4)$] for $n\geq 5$, the total symmetry in $X_1,\ldots,X_{n-3}$.
\end{description}
In actual fact, the permutations $(\mathrm{S}1)$-$(\mathrm{S}4)$
generate all symmetries of the tensors $R_n$'s. For $n=3, 4$, we
have $K_3=K_4=0$ and
\begin{equation*}
R_3=R\,,\qquad R_4=\nabla R\,.
\end{equation*}
In this case,  property $(\mathrm{S}1)$ follows from the standard
antisymmetry of the curvature tensor (\ref{curvature}), and the
properties $(\mathrm{S}2)$, $(\mathrm{S}3)$ reduce to the first
and second Bianchi identities. Therefore one can regards the
properties $(\mathrm{S}2)$, $(\mathrm{S}3)$ as the higher-order
generalization of the Bianchi identities for the curvature tensor.
The explicit calculation of the tensors $K_n$ appears to be quite
a difficult task even for $n=5$ and the complexity  grows rapidly
with $n$. Fortunately, the concrete form of $K_n$'s is absolutely
inessential for our subsequent considerations. What we will
actually use is two facts: (i) the tensors $R_n$ generate the
whole algebra of concomitants associated to the symmetric
connection and (ii) the generators $R_n$ obey no universal
algebraic relations  except for linear relations
$(\mathrm{S}1)$-$(\mathrm{S}4)$.

Applying the Lie derivative $\delta$ to the concomitants
(\ref{R-type}) yields
\begin{equation} \label{dR}
    \begin{array}{l}
    (\delta R_n)(X_1,\ldots ,X_n)=R_{n+1}(Q,X_1,\ldots,X_n)\\[3mm]
\qquad\displaystyle
+\sum_{k=1}^{n}(-1)^{\sum_{i<k}\epsilon(X_i)}R_n(\ldots
,Q_1(X_k),\ldots)-Q_1(R_n(X_1,\ldots ,X_n))+\cdots\,.
    \end{array}
\end{equation}
Here the  dots stand for terms that are at least bilinear in $R_k$'s
with $k<n$. Again, the explicit form of the omitted terms is
inessential for our subsequent calculations.

Taken together the concomitants $\{Q_n\}$ and $\{R_n\}$ constitute
a  multiplicative basis in the differential tensor algebra
$\mathcal{A}$.

\section{The graph complex}\label{GrCom}
The differential tensor algebra of concomitants $(\mathcal{A},
\delta)$ admits a very helpful visualization in terms of finite
graphs with legs. The relevant graphs are composed of black and
white vertices assigned  to the basis concomitants:
\begin{equation}
\begin{split}
\unitlength 1mm 
\begin{picture}(20,20)(-12,-3)
\put(-32,7){$Q_{n}(X_1,\ldots,X_n) \leftrightarrow$}
\put(0,0){\line(1,1){7.5}}\put(4,0){\line(1,2){3.75}}\put(16,0){\line(-1,1){7.5}}
\put(0,0){\vector(1,1){3.75}}\put(4,0){\vector(1,2){1.75}}\put(16,0){\vector(-1,1){3.75}}
\put(8,8){\circle*{1.5}} \put(8,8){\line(0,1){7.5}}
\put(8,8.75){\vector(0,1){3.75}} \put(7,0){{\scriptsize$\ldots$}}
\put(-1.5,-2){{\scriptsize${}_{1}$}}
\put(2.5,-2){{\scriptsize${}_{2}$}}
\put(15.5,-2){{\scriptsize${}_{n}$}}
\end{picture}
\qquad\qquad\qquad{
\unitlength 1mm 
\begin{picture}(20,20)(-20,-3)
\put(-25,7){$R_{n}(X_1,\ldots,X_n) \leftrightarrow$}
\put(7,0){\line(1,1){7.5}}\put(11,0){\line(1,2){3.7}}\put(23,0){\line(-1,1){7.5}}
\put(7,0){\vector(1,1){3.75}}\put(11,0){\vector(1,2){1.75}}\put(23,0){\vector(-1,1){3.75}}
\put(15,8){\circle{1.5}} \put(15,8.75){\line(0,1){7.5}}
\put(15,8.75){\vector(0,1){3.75}}
\put(14,0){{\scriptsize$\ldots$}}
\put(5.5,-2){{\scriptsize${}_{1}$}}
\put(9.5,-2){{\scriptsize${}_{2}$}}
\put(22.5,-2){{\scriptsize${}_{n}$}}
\end{picture}
} \label{basis vertices}
\end{split}
\end{equation}

As is seen the edges incident  to the vertices are
\textit{directed} and each vertex is allowed to have the only
outgoing and several incoming edges. The edges represent,
respectively, the contravariant and covariant tensor indices. The
planar embedding of the vertex graphs  induces the natural
left-to-right \textit{ordering} of the incoming edges. This
ordering, however, is crucial only for the white vertices, since
the basis concomitants corresponding to black vertices are fully
symmetric in covariant indices. Gluing together incoming and
outgoing edges of the vertices above one can produce more general
graphs, which will describe contraction schemes for the tensor
indices of composite concomitants. The graphs we define in such a
way need not be connected and loops are allowed. The remaining
uncontracted tensor indices correspond to \textit{legs}, i.e.,
edges bounded by a vertex from one side and having ``free end'' on
the other. To indicate the order of the indices, the incoming and
outgoing legs are numbered, separately, by consecutive integers
and this defines a \textit{decoration} of a graph.  Permutations
of tensor indices result then in permutations of labels on the
legs. Since the algebra of functions on a supermanifold is
supercommutative rather than commutative in the usual sense,  we
should take into account the sign factors arising upon permutation
of different tensor components in composite concomitants. This
leads us to the concept of \textit{orientation}. By definition, an
orientation on a graph $\Gamma$ is determined by ordering the
vertices of $\Gamma$. Two orientations are the same if they
obtained from one another by a permutation of numbers of black and white
vertices with even number of black-vertex swaps. Finally, to allow
for the integrability condition $Q^2=0$ and its differential
consequences (\ref{QQ}) we exclude from consideration the graphs
that have at least one edge joining a univalent black vertex with
another black vertex. The remaining graphs will be referred to as
$\mathcal{A}$-\textit{graphs}.

\begin{figure}[ht!]
\begin{center}
\subfigure[Antisymmetry in the last two incoming edges]{
\unitlength 1mm 
\begin{picture}(87,20)(-13,-3)
\put(6.75,0){\line(1,1){7.65}}
\put(15,0){\line(0,1){7.25}}\put(19.2,0){\line(-1,2){3.67}}\put(23.5,0){\line(-1,1){7.74}}
\put(6.75,0){\vector(1,1){3.5}}
\put(15,0){\vector(0,1){3.5}}\put(19.2,0){\vector(-1,2){1.75}}\put(23.5,0){\vector(-1,1){3.5}}
\put(15,8){\circle{1.5}} \put(15,8.75){\line(0,1){5.5}}
\put(15,8.75){\vector(0,1){3.75}} \put(9,0){{\scriptsize$\ldots$}}
\put(17.2,-2){{\scriptsize${}_{n\!-\!1}$}}
\put(23.8,-2){{\scriptsize${}_{n}$}} \put(27,7){$+$}
\put(33.75,0){\line(1,1){7.65}}
\put(42,0){\line(0,1){7.25}}\put(46.2,0){\line(-1,2){3.67}}\put(50.5,0){\line(-1,1){7.74}}
\put(33.75,0){\vector(1,1){3.5}}\put(42,0){\vector(0,1){3.5}}\put(46.2,0){\vector(-1,2){1.75}}\put(50.5,0){\vector(-1,1){3.5}}
\put(42,8){\circle{1.5}} \put(42,8.75){\line(0,1){5.5}}
\put(42,8.75){\vector(0,1){3.75}}
\put(36,0){{\scriptsize$\ldots$}}
\put(49,-2){{\scriptsize${}_{n\!-\!1}$}}
\put(45.5,-2){{\scriptsize${}_{n}$}} \put(54.25,7){$=0$}
\end{picture}
}

\subfigure[The first Bianchi identity]{
\unitlength 1mm 
\begin{picture}(75,20)(12,-3)
\put(6.75,0){\line(1,1){7.65}}\put(15,0){\line(0,1){7.25}}\put(19.2,0){\line(-1,2){3.67}}\put(23.5,0){\line(-1,1){7.74}}
\put(6.75,0){\vector(1,1){3.5}}\put(15,0){\vector(0,1){3.5}}\put(19.2,0){\vector(-1,2){1.75}}\put(23.5,0){\vector(-1,1){3.5}}
\put(15,8){\circle{1.5}} \put(15,8.75){\line(0,1){5.5}}
\put(15,8.75){\vector(0,1){3.75}} \put(9,0){{\scriptsize$\ldots$}}
\put(13,-2){{\scriptsize${}_{n\!-\!2}$}}
\put(18.2,-2){{\scriptsize${}_{n\!-\!1}$}}
\put(23.8,-2){{\scriptsize${}_{n}$}} \put(27,7){$+$}
\put(33.75,0){\line(1,1){7.65}}\put(42,0){\line(0,1){7.25}}\put(46.2,0){\line(-1,2){3.67}}\put(50.5,0){\line(-1,1){7.74}}
\put(33.75,0){\vector(1,1){3.5}}\put(42,0){\vector(0,1){3.5}}\put(46.2,0){\vector(-1,2){1.75}}\put(50.5,0){\vector(-1,1){3.5}}
\put(42,8){\circle{1.5}} \put(42,8.75){\line(0,1){5.5}}
\put(42,8.75){\vector(0,1){3.75}}
\put(36,0){{\scriptsize$\ldots$}}
\put(40,-2){{\scriptsize${}_{n\!-\!1}$}}
\put(49.5,-2){{\scriptsize${}_{n\!-\!2}$}}
\put(45.5,-2){{\scriptsize${}_{n}$}} \put(54.25,7){$+$}
\put(60.75,0){\line(1,1){7.65}}\put(69,0){\line(0,1){7.25}}\put(73.2,0){\line(-1,2){3.67}}\put(77.5,0){\line(-1,1){7.74}}
\put(60.75,0){\vector(1,1){3.5}}\put(69,0){\vector(0,1){3.5}}\put(73.2,0){\vector(-1,2){1.75}}\put(77.5,0){\vector(-1,1){3.5}}
\put(69,8){\circle{1.5}} \put(69,8.75){\line(0,1){5.5}}
\put(69,8.75){\vector(0,1){3.75}}
\put(63,0){{\scriptsize$\ldots$}}
\put(77,-2){{\scriptsize${}_{n\!-\!1}$}}
\put(72.2,-2){{\scriptsize${}_{n\!-\!2}$}}
\put(68.5,-2){{\scriptsize${}_{n}$}} \put(81.25,7){$=0$}
\end{picture}
}

\subfigure[The second Bianchi identity]{
\unitlength 1mm 
\begin{picture}(87,20)(5,-3)
\put(3.75,0){\line(4,3){10,5}}\put(11,0){\line(1,2){3.65}}\put(15,0){\line(0,1){7.25}}\put(19,0){\line(-1,2){3.65}}\put(23.5,0){\line(-1,1){7.74}}
\put(3.75,0){\vector(4,3){4.7}}\put(11,0){\vector(1,2){1.75}}\put(15,0){\vector(0,1){3.5}}\put(19,0){\vector(-1,2){1.75}}\put(23.5,0){\vector(-1,1){3.5}}
\put(15,8){\circle{1.5}} \put(15,8.75){\line(0,1){5.5}}
\put(15,8.75){\vector(0,1){3.75}} \put(6,0){{\scriptsize$\ldots$}}
\put(8,-2){{\scriptsize${}_{n\!-\!3}$}}
\put(13.5,-2){{\scriptsize${}_{n\!-\!2}$}}
\put(19,-2){{\scriptsize${}_{n\!-\!1}$}}
\put(24,-2){{\scriptsize${}_{n}$}} \put(27.25,8){$+$}
\put(33.75,0){\line(4,3){10,5}}\put(41,0){\line(1,2){3.65}}\put(45,0){\line(0,1){7.25}}\put(49,0){\line(-1,2){3.65}}\put(53.5,0){\line(-1,1){7.74}}
\put(33.75,0){\vector(4,3){4.7}}\put(41,0){\vector(1,2){1.75}}\put(45,0){\vector(0,1){3.5}}\put(49,0){\vector(-1,2){1.75}}\put(53.5,0){\vector(-1,1){3.5}}
\put(45,8){\circle{1.5}} \put(45,8.75){\line(0,1){5.5}}
\put(45,8.75){\vector(0,1){3.75}}
\put(36,0){{\scriptsize$\ldots$}}
\put(38,-2){{\scriptsize${}_{n\!-\!1}$}}
\put(49,-2){{\scriptsize${}_{n}$}}
\put(52.5,-2){{\scriptsize${}_{n\!-\!3}$}}
\put(43.5,-2){{\scriptsize${}_{n\!-\!2}$}} \put(57.25,8){$+$}
\put(63.75,0){\line(4,3){10,5}}\put(71,0){\line(1,2){3.65}}\put(75,0){\line(0,1){7.25}}\put(79,0){\line(-1,2){3.65}}\put(83.5,0){\line(-1,1){7.74}}
\put(63.75,0){\vector(4,3){4.7}}\put(71,0){\vector(1,2){1.75}}\put(75,0){\vector(0,1){3.5}}\put(79,0){\vector(-1,2){1.75}}\put(83.5,0){\vector(-1,1){3.5}}
\put(75,8){\circle{1.5}} \put(75,8.75){\line(0,1){5.5}}
\put(75,8.75){\vector(0,1){3.75}}
\put(66,0){{\scriptsize$\ldots$}}
\put(78,-2){{\scriptsize${}_{n\!-\!3}$}}
\put(83,-2){{\scriptsize${}_{n\!-\!1}$}}
\put(73,-2){{\scriptsize${}_{n\!-\!2}$}}
\put(70,-2){{\scriptsize${}_{n}$}} \put(87.25,8){$=0$}
\end{picture}
}

\subfigure[Total symmetry in the first $n-3$
    incoming edges]{
\unitlength 1mm 
\begin{picture}(87,20)(-13,-3)
\put(3.75,0){\line(4,3){10,5}}\put(11,0){\line(1,2){3.65}}\put(15,0){\line(0,1){7.25}}\put(23.5,0){\line(-1,1){7.74}}
\put(3.75,0){\vector(4,3){4.7}}\put(11,0){\vector(1,2){1.75}}\put(15,0){\vector(0,1){3.5}}\put(23.5,0){\vector(-1,1){3.5}}
\put(15,8){\circle{1.5}} \put(15,8.75){\line(0,1){5.5}}
\put(15,8.75){\vector(0,1){3.75}} \put(6,0){{\scriptsize$\ldots$}}
\put(17,0){{\scriptsize$\ldots$}} \put(10,-2){{\scriptsize${}_{i}$}}
\put(13.5,-2){{\scriptsize${}_{i\!+\!1}$}}
\put(24,-2){{\scriptsize${}_{n}$}} \put(27.25,8){$-$}
\put(33.75,0){\line(4,3){10,5}}\put(41,0){\line(1,2){3.65}}\put(45,0){\line(0,1){7.25}}\put(53.5,0){\line(-1,1){7.74}}
\put(33.75,0){\vector(4,3){4.7}}\put(41,0){\vector(1,2){1.75}}\put(45,0){\vector(0,1){3.5}}\put(53.5,0){\vector(-1,1){3.5}}
\put(45,8){\circle{1.5}} \put(45,8.75){\line(0,1){5.5}}
\put(45,8.75){\vector(0,1){3.75}} \put(36,0){{\scriptsize$\ldots$}}
\put(47,0){{\scriptsize$\ldots$}}
\put(38.5,-2){{\scriptsize${}_{i\!+\!1}$}}
\put(54,-2){{\scriptsize${}_{n}$}}
\put(45.5,-2){{\scriptsize${}_{i}$}} \put(57.25,8){$=0$}
\end{picture}
}

\caption{Symmetries of white vertices}\label{v-symmetries}
\end{center}
\end{figure}

Now we are ready to define a \textit{graph complex} $(\mathcal{G},
\partial)$ associated to the differential tensor algebra of concomitants
$\mathcal{A}$. The group of $k$-cochains $\mathcal{G}^k$ is, by
definition, a quotient of the real vector space spanned by
$\mathcal{A}$-graphs with $k$ black vertices and arbitrary number
of white vertices:
\begin{equation*}
    \mathcal{G}^k=\mathbb{R}[\mathcal{A}\mbox{\textit{-graphs with k
    black vertices}}]/\mbox{\textit{relations}}\,,
\end{equation*}
where the relations are of two sorts:
\begin{enumerate}
    \item (Orientation) $(\Gamma, -or)=-(\Gamma, or)$.
    \item (Vertex symmetries) The order of edges coming to the black
    vertices is considered to be inessential (total symmetry), while the symmetry of
    white vertices is described by the equivalence relations shown
    in Fig. \ref{v-symmetries}.
\end{enumerate}
Thus, $\mathcal{G}^k$ is spanned by $\mathcal{A}$-graphs with $k$
black vertices and arbitrary number of white vertices. We set
\begin{equation*}
    \mathcal{G}=\bigoplus_{k\geq 0} \mathcal{G}^k\,.
\end{equation*}

\begin{figure}[ht!]
\unitlength 1mm 
\begin{picture}(60,25)(-15,-5)
\put(-9,6){$\partial\left(\rule{0mm}{10mm}\right.$}
\put(17,6){$\left.\rule{0mm}{10mm}\right) = \qquad $
{\large$\sum$}} \put(28,1.75){\scriptsize$I'\sqcup I''=I$}
\put(25,-2.5){\scriptsize$|I'|>0, |I''|>1$}
\put(0,0){\line(1,1){7.5}}\put(4,0){\line(1,2){3.75}}\put(16,0){\line(-1,1){7.5}}
\put(0,0){\vector(1,1){3.75}}\put(4,0){\vector(1,2){1.75}}\put(16,0){\vector(-1,1){3.75}}
\put(8,8){\circle*{1.5}} \put(8,8){\line(0,1){7.5}}
\put(8,8){\vector(0,1){3.75}} \put(7,0){{\scriptsize$\ldots$}}
\put(9,8){{\scriptsize${}_1$}}
\put(-1,-0.5){$\underbrace{\rule{18mm}{0mm}}_{I}$}
\put(47,5){\line(1,1){7.5}}\put(51,5){\line(1,2){3.75}}\put(63,5){\line(-1,1){7.5}}
\put(47,5){\vector(1,1){3.75}}\put(51,5){\vector(1,2){1.75}}\put(63,5){\vector(-1,1){3.75}}
\put(55,13){\circle*{1.5}} \put(55,13){\line(0,1){7.5}}
\put(55,13){\vector(0,1){3.75}} \put(54,5){{\scriptsize$\ldots$}}
\put(56,13){{\scriptsize${}_1$}}
\put(46,4.5){$\underbrace{\rule{12mm}{0mm}}_{I'}$}
\put(55,-3){\line(1,1){7.5}}\put(59,-3){\line(1,2){3.75}}\put(71,-3){\line(-1,1){7.5}}
\put(55,-3){\vector(1,1){3.75}}\put(59,-3){\vector(1,2){1.75}}\put(71,-3){\vector(-1,1){3.75}}
\put(63,5){\circle*{1.5}} \put(62,-3){{\scriptsize$\ldots$}}
\put(64,5){{\scriptsize${}_2$}}
\put(54,-3.5){$\underbrace{\rule{18mm}{0mm}}_{I''}$}
\put(73,6){$,$}
\end{picture}
\qquad\qquad{
\unitlength 1mm 
\begin{picture}(60,25)(-27,-5)
\put(-7,6){$\partial\left(\rule{0mm}{8mm}\right.$}
\put(8,6){$\left.\rule{0mm}{8mm}\right) = 0, $}
\put(1,0){\line(1,2){2.7}}\put(7,0){\line(-1,2){2.7}}
\put(1,0){\vector(1,2){1.75}}\put(7,0){\vector(-1,2){1.75}}
\put(4,6){\circle*{1.5}} \put(4,6){\line(0,1){6.5}}
\put(4,6){\vector(0,1){3.75}}
\end{picture}
}
\begin{picture}(60,20)(-35,0)
\put(-36,8){\line(1,0){8}}\put(-43,8){\line(1,0){8}}
\put(-36,8){\vector(1,0){6}}\put(-43,8){\vector(1,0){4}}
\put(-36,8){\circle*{1.5}} \put(-26,7){$) =$}
\put(-48,7){$\partial\; ($} \put(-36.5,10){{\scriptsize${}_{1}$}}
\put(-10,8){\line(1,0){8}}\put(-2,8){\line(1,0){8}}\put(-18,8){\line(1,0){8}}
\put(-10,8){\vector(1,0){4.5}}\put(-2,8){\vector(1,0){5}}\put(-18,8){\vector(1,0){4}}
\put(-10,8){\circle*{1.5}} \put(-2,8){\circle*{1.5}}
\put(-2.5,10){{\scriptsize${}_{2}$}}
\put(-10.5,10){{\scriptsize${}_{1}$}} \put(8,7){$-\frac12$}
\put(16,8){\line(1,0){7,25}}\put(24.75,8){\line(1,0){8}}\put(21,2){\line(1,2){2.7}}\put(27,2){\line(-1,2){2.7}}
\put(16,8){\vector(1,0){4}}\put(24.75,8){\vector(1,0){5}}\put(21,2){\vector(1,2){1.7}}\put(27,2){\vector(-1,2){1.7}}
\put(24,8){\circle{1.5}}
\put(21,2){\circle*{1.5}}\put(27,2){\circle*{1.5}}
\put(19,0){{\scriptsize${}_{2}$}}
\put(28,0){{\scriptsize${}_{1}$}}
\put(23.5,10){{\scriptsize${}_{3}$}} \put(33.5,7){$,$}
\end{picture}
\qquad\qquad{
\unitlength 1mm 
\begin{picture}(20,10)(-10,-5)
\put(-7.6,2){$\partial\left(\rule{0mm}{5mm} \quad \right)=0,$}
\put(0,0){\circle*{1.5}} \put(0,0){\line(0,1){7}}
\put(0,0){\vector(0,1){4.5}}
\end{picture}
}
\begin{picture}(60,30)(5,-5)
\put(-27,6){$\partial\left(\rule{0mm}{10mm}\right.$}
\put(-1,6){$\left.\rule{0mm}{10mm}\right) = \qquad $}
\put(-18,0){\line(1,1){7.5}}\put(-14,0){\line(1,2){3.7}}\put(-2,0){\line(-1,1){7.5}}
\put(-18,0){\vector(1,1){3.75}}\put(-14,0){\vector(1,2){1.75}}\put(-2,0){\vector(-1,1){3.75}}
\put(-10,8){\circle{1.5}} \put(-10,8.75){\line(0,1){7.5}}
\put(-10,8.75){\vector(0,1){3.75}}
\put(-11,0){{\scriptsize$\ldots$}}
\put(-19.25,-2){{\scriptsize${}_{1}$}}
\put(-15.25,-2){{\scriptsize${}_{2}$}}
\put(-2.75,-2){{\scriptsize${}_{n}$}}
\put(-9,8){{\scriptsize${}_1$}}
\put(14,0){\line(1,1){7.5}}\put(18,0){\line(1,2){3.7}}\put(30,0){\line(-1,1){7.5}}
\put(14,0){\vector(1,1){3.75}}\put(18,0){\vector(1,2){1.75}}\put(30,0){\vector(-1,1){3.75}}
\put(14,0){\circle*{1.5}} \put(22,8){\circle{1.5}}
\put(22,8.75){\line(0,1){7.5}} \put(22,8.75){\vector(0,1){3.75}}
\put(21,0){{\scriptsize$\ldots$}}
\put(16.75,-2){{\scriptsize${}_{1}$}}
\put(29,-2){{\scriptsize${}_{n}$}} \put(23,8){{\scriptsize${}_2$}}
\put(11.5,0.5){{\scriptsize${}_1$}}
\put(33,6){$-$}
\put(39,0){\line(1,1){7.5}}\put(43,0){\line(1,2){3.7}}\put(55,0){\line(-1,1){7.5}}
\put(39,0){\vector(1,1){3.75}}\put(43,0){\vector(1,2){1.75}}\put(55,0){\vector(-1,1){3.75}}
\put(47,8){\circle{1.5}} \put(47,8.75){\line(0,1){5.5}}
\put(47,8.75){\vector(0,1){3.75}} \put(47,14.25){\line(0,1){5.5}}
\put(47,14.25){\vector(0,1){3.75}} \put(47,14.25){\circle*{1.5}}
\put(46,0){{\scriptsize$\ldots$}}
\put(37.75,-2){{\scriptsize${}_{1}$}}
\put(41.75,-2){{\scriptsize${}_{2}$}}
\put(54,-2){{\scriptsize${}_{n}$}} \put(48,8){{\scriptsize${}_1$}}
\put(48,14.25){{\scriptsize${}_2$}}
\put(64,1.75){\scriptsize$k$} \put(58,6){$+ \, $\large$\sum$}
\put(69,0){\line(1,1){8.5}}\put(78,3){\line(0,1){5.25}}\put(87,0){\line(-1,1){8.5}}
\put(69,0){\vector(1,1){3.75}}\put(78,3){\vector(0,1){1}}\put(87,0){\vector(-1,1){3.75}}
\put(78,9){\circle{1.5}} \put(78,9.75){\line(0,1){6.5}}
\put(78,9.75){\vector(0,1){3.75}}
\put(80.5,0){{\scriptsize$\ldots$}}
\put(71.2,0){{\scriptsize$\ldots$}} \put(78,-3){\line(0,1){5.5}}
\put(78,-3){\vector(0,1){3.75}} \put(78,3){\circle*{1.5}}
\put(78.25,-4){{\scriptsize$_k$}}
\put(67.75,-2){{\scriptsize${}_{1}$}}
\put(86,-2){{\scriptsize${}_{n}$}} \put(79,9){{\scriptsize${}_2$}}
\put(79,3.5){{\scriptsize${}_1$}} \put(90,6){$+ \cdots$}
\end{picture}
\caption{ Action of the coboundary operator on the black and white
vertices. In the first equality $|I|\geq 3$. The omitted terms in
the r.h.s. of the last equality are given by graphs with two or
more white vertices.}\label{delta action}
\end{figure}

The coboundary  operator $\partial: \mathcal{G}^k\rightarrow
\mathcal{G}^{k+1}$ is defined to be the \textit{graph}ical
representation of the differential $\delta$ given by Rels.
(\ref{dQ}), (\ref{dR}). Namely, let $\Gamma$ be an
$\mathcal{A}$-graph and let $v$ be a vertex of $\Gamma$. Define an
element $\Gamma_v \in \mathcal{G}$ as follows. If the valency of $v$
is 1 or 3, then $\Gamma_v$ is zero. In the opposite case, $\Gamma_v$
is obtained from $\Gamma$ by replacing the vertex $v$ with a linear
combination of graphs $\partial v$ as shown in Fig. \ref{delta
action}. The orientation on $\Gamma_v$ is determined by the
following rule: choose a representative of the orientation on
$\Gamma$ so that $v$ is the first vertex, then the new vertices of
$\partial v$ are numbered as in Fig. \ref{delta action} and the
numbers of all other vertices increase by 1. With this orientation
convention the map
\begin{equation*}
    \partial \Gamma=\sum_{v\in \Gamma}\Gamma_v
\end{equation*}
is a coboundary operator. Denote the corresponding cohomology
groups by $H(\mathcal{G})=\bigoplus H^k(\mathcal{G})$.

Clearly, the above relation between $\mathcal{A}$-graphs and
concomitants gives rise to an epimorphism $\phi:
\mathcal{G}\rightarrow \mathcal{A}$ of vector spaces. Moreover, the
way we have defined the operator $\partial$ shows  that $\phi$ is a
cochain map. The following definition is central for our
consideration.
\begin{definition}
The stable characteristic classes of $Q$-manifolds are the
$\delta$-cohomology classes belonging to the image of the
homomorphism
\begin{equation}\label{hom}
    H(\phi): H(\mathcal{G})\rightarrow
    H(\mathcal{A})\,.
\end{equation}
\end{definition}

\begin{rem}\label{rem}
One can view each element $g \in \mathcal{G}$ as defining a
(nonlinear) differential operator $\widehat{g}$ on homological
vector fields and symmetric connections with values in tensor
fields. The assignment $g\mapsto \widehat{g}$ defines a
homomorphism $\widehat{\phi}: \mathcal{G}\rightarrow
\frak{Nat}_{Q,\nabla}$ from the space of graphs to the space of
\textit{natural differential operators} \cite{KMS}, \cite{Mar2}
that act on a homological vector field $Q$ and a symmetric
connection $\nabla$. The homomorphism $\phi$ decomposes as
$\phi=\mathrm{ev}_{_{Q,\nabla}}\circ \widehat{{\phi}}$, where
$\mathrm{ev}_{_{Q,\nabla}}$ is the evaluation map. The map
$\widehat{\phi}$ is known to be surjective and so is the map
$\phi$. The kernel of $\mathrm{ev}_{_{Q,\nabla}}$ depends on a
particular form of $Q$ and $\nabla$, while the kernel of
$\widehat{\phi}$ is completely determined by the dimension of the
underlying supermanifold $M$. Namely, adapting the Main Theorem of
Invariant Theory \cite{KMS}, \cite{W} to our situation, one can
argue that the map $\widehat{\phi}$ is an isomorphism in stable
range of dimensions. More precisely, if
$\mathcal{G}^{(m,n)}\subset \mathcal{G}$ is the subspace of graphs
with $m$ vertices and $n$ legs, then the restriction of
$\widehat{\phi}$ to $\mathcal{G}^{(m,n)}$ is a bijection provided
that $\min(s,t)\gg \min(m,n)$, where $(s,t)=\dim M$. Some exact
evolutions of the low bound of stable dimensions can be found in
\cite{Mar2}, \cite{FF}, \cite{Mar1}. Beyond the stable range the
map $H(\phi)$ is neither surjective nor injective. In the
Introduction, the space of stable characteristic classes
$\mathrm{Im}H(\phi)$ was denoted by
$H_{\mathrm{st}}(\mathcal{A})$.
\end{rem}

\begin{rem}
The graph complex $\mathcal{G}$ underlying the definition of
stable characteristic classes is not the only possible or most
natural choice. In Sec. 4.2 we will also consider a quotient
complex $\mathcal{G}/\mathcal{P}$ associated to a certain
subcomplex $\mathcal{P}\subset \mathcal{G}$. The homomorphism
$\phi$ passes through the quotient for an appropriate choice of
symmetric connection, giving rise to an additional series of
characteristic classes ($A$-series). It is the $A$-series of
invariants of $Q$-manifolds that was originally discovered in
\cite{LS} and called the principal series of characteristic
classes.
\end{rem}

The rest of the paper is devoted to computation of the graph
cohomology.

Since $\partial$ does not affect the legs of $\mathcal{A}$-graphs,
the graph complex is decomposed into the direct sum of
subcomplexes
\begin{equation*}
\mathcal{G}=\bigoplus_{n,m} {\mathcal{G}}_{n,m}\,,
\end{equation*}
where the subscripts $n$ and $m$ refer to the number of incoming
and outgoing legs of graphs. Another  remarkable property of the
coboundary operator $\partial$ is that it neither permutes the
connected components of an $\mathcal{A}$-graph nor changes their
number.  This leads to the further decomposition of
$\mathcal{G}_{n,m}$ into the direct sum of subcomplexes
$\mathcal{G}_{n,m}^{I_1,...,I_k;J_1,...,J_k}$, where
$\{I_1,...,I_k\}$ and $\{J_1,...,J_k\}$ are partitions of the sets
$\{1,...,n\}$ and $\{1,...,m\}$. The complex
$\mathcal{G}_{n,m}^{I_1,...,I_k;J_1,...,J_k}$ is generated by
graphs with $k$ connected components such that the incoming and
outgoing legs of the $l$th component are labelled by the element
of $I_l$ and $J_l$, respectively. Notice that the sets
$\{I_1,...,I_k\}$ and $\{J_1,...,J_k\}$ are defined up to
simultaneous permutations of $I_q$ with $I_p$ and $J_q$ with
$J_p$, and some of the sets $I_1,...,I_k$, $J_1,...,J_k$ may be
empty. Let $\bar{\mathcal{G}}=\bigoplus\bar{\mathcal{G}}_{n,m}$
denote the subcomplex of \textit{connected} graphs. It is clear
that
\begin{equation*}
\mathcal{G}_{n,m}^{I_1,...,I_k;J_1,...,J_k}\simeq
\bigotimes_{l=1}^k \bar{\mathcal{G}}_{|I_l|,|J_l|}\,,
\end{equation*}
and by the K\"unneth formula the computation of the graph
cohomology boils down  to the computation of the groups
$H^k(\bar{\mathcal{G}}_{n,m})$.

The characteristic classes that belong to the image of the
connected graph cohomology under the map (\ref{hom}) will be
called \textit{primitive}. A linear basis in the space of all
characteristic classes is made up of the primitive characteristic
classes by means of tensor products and permutations of tensor
indices.

\section{The cohomology of the connected graph complex}\label{5}

We start with the  observation that the complex of connected
graphs $\bar{\mathcal{G}}$ splits into a direct sum of four
subcomplexes,
\begin{equation*}
\bar{\mathcal{G}}=\bar{\mathcal{G}}^{(1)}\oplus
\bar{\mathcal{G}}^{(2)}\oplus \bar{\mathcal{G}}^{(3)}\oplus
\bar{\mathcal{G}}^{(4)}\,.
\end{equation*}
Here $\bar{\mathcal{G}}^{(1)}$ is the one dimensional complex
spanned by the graph $\bullet \!\!\!\!\rightarrow$. To describe
the complex $\bar{\mathcal{G}}^{(2)}$ it is convenient to
introduce a special notation\footnote{This cannot cause a
confusion since the basis white vertices (\ref{basis vertices})
have valency $\geq 4$.} for a graph entering the third relation in
Fig. \ref{delta action}, namely,
\begin{equation}
\begin{split}
\unitlength 1mm 
\begin{picture}(60,13)(-30,1)
\put(-22,7){$ \frac12 $}
\put(2,7){$ \equiv \qquad $}
\put(-18,8){\line(1,0){7.25}}\put(-9.25,8){\line(1,0){8}}\put(-13,2){\line(1,2){2.7}}\put(-7,2){\line(-1,2){2.7}}
\put(-18,8){\vector(1,0){4}}\put(-9.25,8){\vector(1,0){5}}\put(-13,2){\vector(1,2){1.7}}\put(-7,2){\vector(-1,2){1.7}}
\put(-10,8){\circle{1.5}}
\put(-13,2){\circle*{1.5}}\put(-7,2){\circle*{1.5}}
\put(-15,0){{\scriptsize${}_{2}$}}
\put(-6,0){{\scriptsize${}_{1}$}}
\put(8,8){\line(1,0){7.25}}\put(16.75,8){\line(1,0){8}}
\put(8,8){\vector(1,0){4}}\put(16.75,8){\vector(1,0){5}}
\put(16,8){\circle{1.5}}\put(25.5,7){$.$}
\end{picture}\label{bivalent}
\end{split}
\end{equation}
 The action of the differential on
the bivalent white vertex reads
\begin{equation*}
\unitlength 1mm 
\begin{picture}(60,10)(-10,5)
\put(-23,7){$\partial\left(\rule{0mm}{3mm}\right.$}
\put(-1,7){$\left.\rule{0mm}{3mm}\right) = \qquad $}
\put(-18,8){\line(1,0){7.25}}\put(-9.25,8){\line(1,0){8}}
\put(-18,8){\vector(1,0){4}}\put(-9.25,8){\vector(1,0){5}}
\put(-10,8){\circle{1.5}} \put(-10.5,10){{\scriptsize${}_{1}$}}
\put(8,8){\line(1,0){7.25}} \put(8,8){\vector(1,0){4}}
\put(16,8){\circle*{1.5}}
\put(16,8){\line(1,0){5.25}}\put(22.75,8){\line(1,0){8}}
\put(16,8){\vector(1,0){4}}\put(22.75,8){\vector(1,0){5}}
\put(22,8){\circle{1.5}} \put(15.5,10){{\scriptsize${}_{1}$}}
\put(21.5,10){{\scriptsize${}_{2}$}}
\put(33,7){$-$}
\put(39,8){\line(1,0){7.25}}\put(47.75,8){\line(1,0){8}}
\put(39,8){\vector(1,0){4}}\put(47.75,8){\vector(1,0){3.5}}
\put(47,8){\circle{1.5}} \put(52.75,8){\line(1,0){7.25}}
\put(52.75,8){\vector(1,0){4}} \put(52.75,8){\circle*{1.5}}
\put(52,10){{\scriptsize${}_{2}$}}
\put(46.5,10){{\scriptsize${}_{1}$}}\put(60.85,7){$.$}
\end{picture}
\end{equation*}
As is seen the vector space of graphs composed of black and white
bivalent vertices is invariant under the action of $\partial$ and
we identify this space with $\bar{\mathcal{G}}^{(2)}$. The graphs
generating $\bar{\mathcal{G}}^{(3)}$ contain only black vertices
of valency $\geq 3$. Finally, the linear span of all other
$\mathcal{A}$-graphs  defines the complex
$\bar{\mathcal{G}}^{(4)}$.

\subsection{The cohomology of $\bar{\mathcal{G}}^{(1)}$} It is clear that
 $H(\bar{\mathcal{G}}^{(1)})\cong\bar{\mathcal{G}}^{(1)}\cong \mathbb{R}$.

\subsection{The cohomology of $\bar{\mathcal{G}}^{(2)}$} It was shown
in \cite{LMS2} that the cohomology of $\bar{\mathcal{G}}^{(2)}$ is
trivial. This a little bit disappointing fact means that we have
no nontrivial characteristic classes associated to the graph
complex $\bar{\mathcal{G}}^{(2)}$. The situation, however, is not
so hopeless as might appear. Observe that the complex
$\bar{\mathcal{G}}^{(2)}$ contains the subcomplex $\mathcal{P}$
spanned by the cocycles

\begin{equation*}
\begin{split}
\unitlength 1.5mm
\begin{picture}(60,10)(10,0)
\thinlines \curvedashes[1mm]{0,7,2} \put(4,4.5)
{$\Pi_{2n-1}=-\binom{4n-3}{2n-1}$}
\put(37,4.5){$,\qquad\partial\Pi_{2n-1}=0\,,\qquad \forall n\in
\mathbb{N}\,.$}
 \put(32.5,0.7){\vector(2,1){0.5}}
\put(29.8,0){\circle{1.15}} \put(29.8,-2){{\scriptsize$2$}}
\put(26.5,1.5){\vector(1,-1){0.5}} \put(25.2,3.4){\circle{1.15}}
\put(23.7,1.5){{\scriptsize$1$}}
\put(25.6,7.2){\vector(-1,-2){0.5}} \put(26.7,8.8){\circle{1.15}}
\put(23.5,10.5){{\scriptsize$2n\!-\!1$}}
\put(29.8,9.95){\vector(-1,0){0.5}} \put(32.4,9.4){\circle{1.15}}
\put(33,10.5){{\scriptsize$2n\!-\!2$}}
\put(34.7,6.7){\vector(-1,2){0.5}} \put(20,5){
\scaleput(10,0){\arc(5,0){320}} \curvedashes[1mm]{0,0.3,0.9}
\scaleput(10,0){\arc(5,0){-40}} }
\end{picture}
\end{split}
\end{equation*}
Let $\pi:\bar{\mathcal{G}}^{(2)}\rightarrow
\bar{\mathcal{G}}^{(2)}/\mathcal{P}$ denote the canonical
projection.  The acyclicity of $\bar{\mathcal{G}}^{(2)}$ implies
two facts: (i) there exists a cochain $\Psi_{2n-1}\in
\bar{\mathcal{G}}^{(2)}$ such that  $\Pi_{2n-1}=\partial
\Psi_{2n-1} $ and (ii)  $\Psi'_{2n-1}=\pi(\Psi_{2n-1})$ is a
nontrivial cocycle of $\bar{\mathcal{G}}^{(2)}/\mathcal{P}$. The
relative cocycles $\Psi_{2n-1}$ admit an explicit description
\cite{LS}. Namely, each $\Psi_{2n-1}$ is represented by a linear
combination of cyclic graphs
\begin{equation}\label{A}
\Psi_{2n-1}=\psi_1+\psi_2+\ldots+\psi_{2n-1}
\end{equation}
composed of the black and white bivalent vertices, where the graph

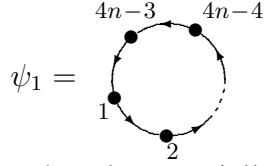
\begin{equation*}
\begin{split}
\unitlength 1.5mm
\begin{picture}(55,10)
\thinlines \curvedashes[1mm]{0,7,2} \put(16,4.5){$\psi_1=$}
\put(32.5,0.7){\vector(2,1){0.5}} \put(29.8,0){\circle*{1.15}}
\put(29.8,-2){{\scriptsize$2$}} \put(26.5,1.5){\vector(1,-1){0.5}}
\put(25.2,3.4){\circle*{1.15}} \put(23.7,1.5){{\scriptsize$1$}}%
\put(25.6,7.2){\vector(-1,-2){0.5}} \put(26.7,8.8){\circle*{1.15}}
\put(23.5,10.5){{\scriptsize$4n\!-\!3$}}
\put(29.8,9.95){\vector(-1,0){0.5}} \put(32.4,9.4){\circle*{1.15}}
\put(33,10.5){{\scriptsize$4n\!-\!4$}}
\put(34.7,6.7){\vector(-1,2){0.5}}
\put(20,5){%
\scaleput(10,0){\arc(5,0){320}} \curvedashes[1mm]{0,0.3,0.9}
\scaleput(10,0){\arc(5,0){-40}} }
\end{picture}
\end{split}
\end{equation*}
consists of $4n-3$ black vertices numbered sequentially and the
other $2n-2$ terms in (\ref{A}) are obtained from each other by
the successive action of three linear operators:
\begin{equation*}
\psi_m=C B A (\psi_{m-1})\,.
\end{equation*}
To describe the  action of the operators $A$, $B$, and $C$ we need
the following terminology. The white vertices divide a cyclic
graph $\Gamma$ into several \textit{arcs} endowed with black
vertices. The \textit{length} of an arc is, by definition, the
number of its black vertices.
\begin{description}
    \item[A] The operator $A$ acts successively on the black
    vertices of $\Gamma$ by inverting their color and  multiplying the result by $(-1)^l$,
     where $l$ is the number of black vertices preceding the
     inverted one. The labels on the vertices remain the same.
     Adding all such graphs up one gets $A(\Gamma)$.

    \item[B] The operator $B$ multiplies the graph $\Gamma$ by $1/k$, where $k$ is the number of arcs of nonzero length.

    \item[C] The operator $C$ acts successively on the nonzero-length arcs of
    $\Gamma$ by decreasing their length by 1. More precisely,
    it removes the first black vertex of an arc and multiplies the
    result by $(-1)^l$,
    where $l$ is the number of black vertices preceding the
    removed one.  If the removed vertex was labelled by $k$, then the labels less than $k$ remain intact while the labels greater than $k$
    are shifted down by 1.
    By summing over all nonzero-length arcs of
    $\Gamma$ one gets $C(\Gamma)$.
\end{description}

Let us now explain the relevance of the cocycles $\Psi'_{2n-1}\in
\bar{\mathcal{G}}^{(2)}/\mathcal{P}$ to the characteristic classes
of $Q$-manifolds. Recall that to any symmetric connection $\nabla$
one can associate a sequence of closed $2m$-forms
    \begin{equation*}
{P}_m=\mathrm{Str}(R^m)\in \Omega^{2m}(M)
\end{equation*}
constructed by the curvature tensor $R\in
\Omega^2(M)\otimes\frak{A}(M)$ of $\nabla$. The de Rham cohomology
class of ${P}_m$ is known as the $m$th Pontryagin character of
$M$. On any supermanifold, there exists a special symmetric
connection $\nabla$ with the property that ${P}_{2n-1}=0$ for all
$n\in \mathbb{N}$ (see e.g. \cite{LMS2}). This results in
triviality of all Pontryagin's characters with odd $m$'s, while
the other characters may well be nontrivial.  Denote by
$\mathcal{A}'$ the algebra of concomitants constructed by this
special connection. It is clear that
\begin{equation*}
\begin{array}{c}
    \phi(\Pi_{2n-1})=\binom{4n-3}{2n-1}{P}_{2n-1}(Q^{\otimes(4n-2)})\,.
    \end{array}
\end{equation*}
Since the last expression vanishes in $\mathcal{A}'$,  we have a
well-defined  cochain map $\phi:
{\mathcal{G}}/\mathcal{P}\rightarrow \mathcal{A}'$, which takes
$\Psi_{2n-1}'$ to the  $\delta$-cocycle
${A}_{2n-1}=\phi(\Psi'_{2n-1})\in \mathcal{A}'$. Following the
prescriptions above, one can easily write down explicit
expressions for the scalar concomitants ${A}_{2n-1}$ with small
$n$'s. To do this in a compact way, it is convenient to identify
the covariant derivative  $Q_1=\nabla Q \in \mathcal{A}'$ with the
right endomorphism $\Lambda \in \frak{A}(M)$ defined by the rule:
$\Lambda(X)=\nabla_X Q$ for all $X\in \frak{X}(M)$. Similarly,
define the right endomorphism $\mathrm{R}\in \frak{A}(M)$ by
setting $\mathrm{R}(X)=\frac12 R(X,Q,Q)$. With these definitions
we have
\begin{equation*}
\begin{array}{l}
{A}_{1} = \mathrm{Str}(\Lambda)\,, \\[3mm]
{A}_{3} = \mathrm{Str}(\Lambda^{5} + 5 \mathrm{R}\Lambda^{3} + 10 \mathrm{R}^{2} \Lambda)\,,  \\[3mm]
{A}_{5} = \mathrm{Str}(\Lambda^{9} + 9 \mathrm{R}\Lambda^{7} + 18
\mathrm{R}^{2} \Lambda^{5}
+ 9 \mathrm{R} \Lambda \mathrm{R} \Lambda^{4} + 9 \mathrm{R} \Lambda^2 \mathrm{R} \Lambda^{3} + \\[3mm]
\qquad \ \ 45 \mathrm{R}^3 \Lambda^3  + 21 \mathrm{R}^2 \Lambda \mathrm{R} \Lambda^{2} + 15 \mathrm{R}^2 \Lambda^2 \mathrm{R} \Lambda + 3 \mathrm{R} \Lambda \mathrm{R} \Lambda \mathrm{R} \Lambda + 126\mathrm{R}^{4} \Lambda)\,. \\[3mm]
\end{array}
\end{equation*}
The infinite sequence of $\delta$-cohomology classes $[A_{2n-1}]$
is called the $A$-series of characteristic classes \cite{LMS2}. In
the particular case of homological vector fields on $NQ$-manifolds
of degree 1 (see e.g.  \cite{Vaintrob1}, \cite{Roy}, \cite{Sev})
the corresponding characteristic classes of $A$-series add up to
the secondary characteristic classes of Lie algebroids \cite{F1},
\cite{F2}.

\subsection{The cohomology of $\bar{\mathcal{G}}^{(3)}$} It was found in
\cite{LMS2} that the computation of $H(\bar{\mathcal{G}}^{(3)})$
is essentially equivalent to the computation of the stable
cohomologies of a  Lie algebra of formal vector fields with tensor
coefficients. The latter was done by Fuks in \cite{Fu}. Before
stating the result let us note that the graphs spanning $
\bar{\mathcal{G}}^{(3)}$ may have at most one outgoing leg. Since
the coboundary operator does not mix graphs with different numbers
of incoming and outgoing legs, we have the following
decomposition:
\begin{equation}\label{sum}
    H(\bar{\mathcal{G}}^{(3)})=
    \left(\bigoplus_{q,n}
    H^q(\bar{\mathcal{G}}_{1,n}^{(3)})\right)\oplus \left(\bigoplus_{p,m} H^p(\bar{\mathcal{G}}_{0,m}^{(3)})\right)\,.
\end{equation}
The space $\bar{\mathcal{G}}^{(3)}_{1,n}$ is spanned by tree
connected graphs with $n$ incoming and 1 outgoing legs, while the
graphs from $\bar{\mathcal{G}}^{(3)}_{0,m}$ contain exactly one
cycle, $m$ incoming and no outgoing legs.

\begin{theorem} All nontrivial groups in the
sum (\ref{sum}) have the following dimensions:
\begin{equation*}
\dim H^{n-1}(\bar{\mathcal{G}}_{1,n}^{(3)})=(n-1)!\,,\qquad \dim
H^n(\bar{\mathcal{G}}^{(3)}_{0,n})=(n-1)!\,.
\end{equation*}
\end{theorem}
As a basis of the nontrivial cocycles one can take the graphs
depicted in Fig. \ref{B&C}. By definition, the $B$-series of
characteristic classes is spanned by the tree graphs with  $n$
incoming legs and $n-1$ trivalent vertices. The left-most leg of these
graphs is labelled by 1 and the labels on the other $n-1$ incoming
legs may  be chosen arbitrary. The total number of different
decorated graphs of type $B_n$ is thus  $(n-1)!$. The graphs of
$C$-series have the form of a cycle composed of $n$ trivalent
vertices with $n$ incoming legs. Since the cyclic permutations of
labels on the incoming legs preserve the isomorphism class of a
decorated graph $C_n$, there are exactly $(n-1)!$ different
decorations.

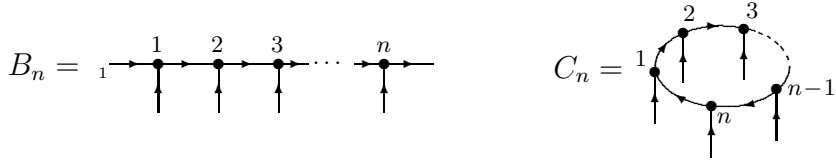
\begin{figure}[ht!]
\begin{center}
{
\unitlength 1mm 
\begin{picture}(75,15)
\put(-20,-1.5){$B_n=$} \put(-6.5,0){\vector(1,0){4}}
\put(-6.5,0){\line(1,0){6.5}} \put(-8,-1.5){{\scriptsize${}_1$}}
\put(0,-6.5){\line(0,1){6.5}} \put(0,-6.5){\vector(0,1){4}}
\put(0,0){\circle*{1.5}} \put(-1,1.5){{\scriptsize$1$}}
\put(0,0){\line(1,0){8}} \put(0,0){\vector(1,0){4.5}}
\put(8,-6.5){\line(0,1){6.5}} \put(8,-6.5){\vector(0,1){4}}
    \put(8,0){\circle*{1.5}}
\put(7,1.5){{\scriptsize$2$}} \put(8,0){\line(1,0){8}}
\put(8,0){\vector(1,0){4.5}} \put(16,-6.5){\line(0,1){6.5}}
\put(16,-6.5){\vector(0,1){4}}
    \put(16,0){\circle*{1.5}}
\put(15,1.5){{\scriptsize$3$}} \put(16,0){\line(1,0){4}}
\put(16,0){\vector(1,0){3}} \put(20.5,0){{\scriptsize$\ldots$}}
\put(26,0){\line(1,0){4}} \put(26,0){\vector(1,0){2.5}}
    \put(30,0){\circle*{1.5}}
\put(29,1.5){{\scriptsize$n$}} \put(30,-6.5){\line(0,1){6.5}}
\put(30,-6.5){\vector(0,1){4}} \put(30,0){\line(1,0){6.5}}
\put(30,0){\vector(1,0){4}}
\end{picture}
}

{
\unitlength 1.5mm
\begin{picture}(0,8)
 \thinlines
\put(40,2){%
\renewcommand{\xscale}{0.65}
\renewcommand{\xscaley}{-1}
\renewcommand{\yscale}{0.4}
\renewcommand{\yscalex}{0.6}
\scaleput(0,15){\arc(4.99,0.35){300}}
\curvedashes[1mm]{0,0.3,0.9}
\scaleput(0,15){\arc(5,0){-55}}%
}
\put(10,7){$C_n=$}
\put(19,7.5){\circle*{1.05}}
\put(17.25,8){{\scriptsize$1$}} \put(19,2.95){\line(0,1){4.55}} \put(19,2.95){\vector(0,1){2.8}}
\put(20,10){\vector(3,4){0.1}}
\put(21.5,11){\circle*{1.05}}
\put(21.5,12){{\scriptsize$2$}} \put(21.5,6.55){\line(0,1){4.55}} \put(21.5,6.55){\vector(0,1){2.8}}
\put(24.5,11.55){\vector(1,0){0.1}}
\put(24,4.5){\circle*{1.05}}
\put(24.5,3){{\scriptsize$n$}} \put(24,-0.05){\line(0,1){4.55}} \put(24,0.05){\vector(0,1){2.8}}
\put(27,4.55){\vector(-4,-1){0.1}}
\put(26.9,11.35){\circle*{1.05}}
\put(27.1,12.25){{\scriptsize$3$}} \put(26.9,6.8){\line(0,1){4.25}} \put(26.9,6.8){\vector(0,1){2.8}}
\put(29.8,6){\circle*{1.05}}
\put(30.75,5.5){{\scriptsize$n\!-\!1$}} \put(29.8,1.45){\line(0,1){4.55}} \put(29.8,1.45){\vector(0,1){2.8}}
\put(20.9,5.5){\vector(-3,2){0.1}}
\end{picture}
} \caption{Basis cocycles of series $B$ and $C$.}\label{B&C}
\end{center}
\end{figure}

The closedness of the graphs $B_n$ and $C_n$ readily  follows from
the fact that the coboundary operator  $\partial$ annihilates the
trivalent black vertex as is seen from Fig. \ref{delta action}. So
any trivalent graph is a cocycle. Interestingly enough the
vertices of valency higher than 3 do not contribute to the
cohomology. In analytical terms this means that the corresponding
tensor cocycles of ${\mathcal{A}}$ can be chosen to involve no
more than second covariant derivatives of the homological vector
field. Using the definition of the basis concomitants (\ref{basis
vertices}), one can easily assign the tensor expressions for the
graphs in Fig. \ref{B&C}. Let us interpret the basis generator
$Q_2$, which corresponds to the trivalent black vertex, as a
$C^{\infty}(M)$-module homomorphism $\frak{X}(M) \rightarrow
\frak{A}(M)$ that takes a vector field $X$ to the right
endomorphism $Q_2(X)$:
\begin{equation*}
Q_2(X)(Y)=Q_2(X,Y)\qquad \forall Y\in
\frak{X}(M)\,.
\end{equation*}
Now, identifying the $(1,n+1)$-tensors with the homomorphisms
$\frak{X}(M)^{\otimes n}\rightarrow \frak{A}(M)$, we can
write\footnote{By abuse of notation, we use the same symbol for a
graph cocycle and its image under the map (\ref{hom}).}
\begin{equation*}
    B_n(X_1,X_2,\ldots,X_n)=(-1)^{\sum_{k}\epsilon(X_{2k})}Q_2(X_n)Q_2(X_{n-1})\cdots Q_2(X_1)\,.
\end{equation*}
The concomitants $C_n$ are then defined by
\begin{equation*}
    C_n(X_1,X_2,\ldots,X_n)=\mathrm{Str}B_n(X_1,X_2,\ldots,X_n)\,.
\end{equation*}
Permuting the arguments $X_1$, $X_2$,\ldots ,$X_n$ one gets the
basis of nontrivial $\delta$-cocycles.

\subsection{The cohomology of $\bar{\mathcal{G}}^{(4)}$} Below we prove that the complex  $\bar{\mathcal{G}}^{(4)}$
is acyclic and this gives  the main result of the paper.

Let us introduce the following terminology.  By
\textit{multivalent vertices} we will mean the black and white
vertices of valency $\geq 3$. Besides multivalent vertices, each
graph $\Gamma\in \bar{\mathcal{G}}^{(4)}$ is allowed to have some
number of univalent and bivalent black vertices. A \textit{branch}
is, by definition, a connected subgraph of $\Gamma$ given by a
maximal string of bivalent vertices bounded by one or two
multivalent vertices. Graphically, a typical branch looks like
\begin{align}\label{branch}
\begin{split}
\unitlength 0.9mm 
\begin{picture}(60,10)(-25,-5)
\put(-32,0){$\gamma_k^{\alpha\beta}=$}
\put(-4.85,1){\line(1,0){8.15}} \put(-4,1){\vector(1,0){3}} \put(11,1){\line(1,0){7.85}} \put(12,1){\vector(1,0){4.5}}
\put(-11,-5){\line(3,4){3.9}}\put(-7.85,-5){\line(1,3){1.55}}\put(-1,-5){\line(-3,4){3.9}}
\put(-11,-5){\vector(3,4){2.6}}\put(-7.85,-5){\vector(1,3){1.1}}\put(-1,-5){\vector(-3,4){2.6}}
\put(-6,-4.5){$_{\ldots}$}
\put(-6,1){\circle{2.2}}
\put(-7.25,0.25){\scriptsize$\times$}
\put(-7.25,4){$_\alpha$}
\put(1,1){\circle*{1.5}}
\put(13,1){\circle*{1.5}}
\put(4.5,0.75){$\ldots$}
\put(20,1){\circle{2.2}}
\put(18.75,0.25){\scriptsize$\times$}
\put(19.25,4){$_\beta$}
\put(15,-5){\line(3,4){3.9}}\put(18.25,-5){\line(1,3){1.55}}\put(25,-5){\line(-3,4){3.9}}
\put(15,-5){\vector(3,4){2.6}}\put(18.25,-5){\vector(1,3){1.1}}\put(25,-5){\vector(-3,4){2.6}}
\put(20,-4.5){$_{\ldots}$}
\put(.5,3.5){{\scriptsize$_1$}}
\put(12.5,3.5){{\scriptsize$_{k}$}}
\put(21.15,1){\line(1,0){7}}
\put(22,1){\vector(1,0){4}}
\end{picture}
\end{split}
\end{align}
The marks $\alpha,\beta \in \{\bullet,\circ,\varnothing\}$ denote
three possible types of boundary vertices (the symbol
$\varnothing$ is used to indicate that the branch ends with a
leg).

\begin{lemma} For any graph of $\bar{\mathcal{G}}^{(4)}$ at least one of the following statements is true:
\begin{enumerate}
    \item there is a branch of nonzero length;
    \item there is an edge joining black and white multivalent vertices;
    \item there is a leg adjacent to a white vertex.
\end{enumerate}\label{lemma}
\end{lemma}
\begin{proof}
Suppose a connected graph  $\Gamma\in \bar{\mathcal{G}}^{(4)}$
fails to meet the first and second  conditions. Then, there are
two options: either $\Gamma$ involves only black vertices  or it
consists of white multivalent and black univalent vertices. In the
former case $\Gamma$ must belong to $\bar{\mathcal{G}}^{(3)}$ and
not to $\bar{\mathcal{G}}^{(4)}$. (In the absence of white
vertices, black univalent vertices cannot coexist with black
multivalent ones in a connected $\mathcal{A}$-graph.) Turning to
the second possibility, we note that the symmetries  of white
vertices (see Fig. \ref{v-symmetries}) lead to  the identities
\begin{equation*}
\begin{split}
\unitlength 1mm 
\begin{picture}(50,15)(0,0)
\put(11,0){\line(1,1){7.5}}\put(19,0){\line(0,1){7.25}}\put(27,0){\line(-1,1){7.5}}
\put(11,0){\vector(1,1){3.75}}\put(19,0){\vector(0,1){3.75}}\put(27,0){\vector(-1,1){3.75}}
\put(19,8){\circle{1.5}} \put(19,8.75){\line(0,1){5.5}}
\put(19,8.75){\vector(0,1){3.75}} \put(11,0){\circle*{1.5}}
\put(19,0){\circle*{1.5}} \put(27,0){\circle*{1.5}}
\put(13,0){{\scriptsize$\ldots$}}
\put(21,0){{\scriptsize$\ldots$}} \put(6,0){{\scriptsize$\ldots$}}
\put(29,0){{\scriptsize$\ldots$}} \put(30,8){$=0\,,$}
\end{picture}
\end{split}
\end{equation*}
where the dots stand for other possible incoming edges/legs.   As
a consequence no more than two univalent black vertices may join
with a white vertex. Having no legs, the graph $\Gamma$ must take
the form of a cyclic graph composed of the bivalent white vertices
(\ref{bivalent}). It remains to note that all such graphs belong
to $\bar{\mathcal{G}}^{(2)}$ rather than
$\bar{\mathcal{G}}^{(4)}$.
\end{proof}

\begin{theorem}
The complex $\bar{\mathcal{G}}^{(4)}$ is acyclic.
\end{theorem}

\begin{proof} The complex $\bar{\mathcal{G}}^{(4)}$ admits a decreasing filtration
\begin{equation*}
    \bar{\mathcal{G}}^{(4)}= F_1\bar{\mathcal{G}}^{(4)} \supset F_2\bar{\mathcal{G}}^{(4)} \supset\cdots\supset
    F_\infty\bar{\mathcal{G}}^{(4)}=0\,,
\end{equation*}
where $F_k\bar{\mathcal{G}}^{(4)}$ spans the graphs with $k$ and
more multivalent and univalent vertices. Define the corresponding
spectral sequence  $\{E_r,d_r\}$. The zero differential $d_0$
increases the number of bivalent vertices, leaving the other
vertices intact. More precisely, if we prescribe the boundary
vertices of the branch (\ref{branch}) the degrees
$|\bullet|=|\varnothing|=0$ and $|\circ|=1$, then
\begin{equation}\label{d0}
d_0 \gamma_k^{\alpha\beta}=
\frac12(\!-1\!)^{|\alpha|}\!(1\!-\!(\!-1\!)^{k+|\alpha|+|\beta|})\gamma_{k+1}^{\alpha\beta}\,.
\end{equation}

We claim  that $E_1=\mathrm{Ker}d_0/\mathrm{Im}d_0=0$. To show
this define the operator $h:\bar{\mathcal{G}}^{(4)}\rightarrow
\bar{\mathcal{G}}^{(4)}$ that acts on the branches of the graph
$\Gamma$, one at a time, so that the result $h\Gamma$ is given by
a signed sum of graphs over all branches of $\Gamma$. To describe
the action of $h$ on an individual branch of $\Gamma$ we may
assume this branch to be oriented as in (\ref{branch}). (The
general case reduces to that by reordering the vertices of
$\Gamma$ with account of sign factor.) Then
\begin{equation}\label{h}
h \gamma_k^{\alpha\beta}=
\frac12(\!-1\!)^{|\alpha|}\!(1\!+\!(\!-1\!)^{k+|\alpha|+|\beta|})\gamma_{k-1}^{\alpha\beta}\,,
\end{equation}
where we assume that $h \gamma_0^{\alpha\beta}=0$. In other words,
the operator  $h$ either annihilates the branch or shortens it by
one bivalent vertex.  The $k-1$ bivalent vertices of the resulted
(nonzero) branch appear to be numbered in order and the labels on
the other vertices of the graph $\Gamma$ are reduced by one. In
this way the new graph gets an orientation.

It follows form relations (\ref{d0}) and (\ref{h}) that the
operator $\Delta=h d_0+d_0h$ is diagonal in the natural basis of
$\bar{\mathcal{G}}^{(4)}$. More precisely, $\Delta
\Gamma=(n_1+n_2+n_3)\Gamma$, where the eigenvalue depends on the
structure of the graph $\Gamma\in \bar{\mathcal{G}}^{(4)}$.
Namely, $n_1$ is the number of branches of nonzero length, $n_2$
is the number of edges joining black and white multivalent
vertices, and $n_3$ is the number of legs adjacent to white
vertices. By Lemma \ref{lemma}, the sum $n_1+n_2+n_3$ is strictly
positive and the operator $\Delta$ is invertible. Taking the
composition $h\Delta^{-1}$ as a contracting homotopy for $d_0$, we
see that the complex $(E_0,d_0)$ is acyclic and so is the graded
complex associated to the filtered complex
$\bar{\mathcal{G}}^{(4)}$. On the other hand,
\begin{equation}\label{prlim}
\bar{\mathcal{G}}^{(4)}=\lim_{k}\bar{\mathcal{G}}^{(4)}/F_k\bar{\mathcal{G}}^{(4)}\,,
\end{equation}
 where $\lim$ denotes the projective limit. (The last equality follows from
the fact that we consider graphs with finite number of edges and
vertices.) It remains  to note that acyclicity of the associated
graded complex of $\bar{\mathcal{G}}^{(4)}$ implies acyclicity of
the r.h.s. of (\ref{prlim}).
\end{proof}

We close this section with two theorems. The first theorem was
proved in \cite{LMS2}, while the second one summarizes the results
of the present paper.

\begin{theorem}
The characteristic classes of homological vector fields are
independent of the choice of symmetric connection.
\end{theorem}

\begin{theorem}
All the primitive characteristic classes of a homological vector
field $Q$ are grouped into the two infinite series $B$ and $C$
plus the $\delta$-cohomology class $[Q]$. For a special choice of
symmetric connection one can also  define the $A$-series of
characteristic classes.
\end{theorem}

\subsection{Example} To show nontriviality of the constructed
characteristic classes we consider the homological vector field
associated to a Lie algebra $\mathcal{L}$. Let $\{t_a\}$ be a
basis in $\mathcal{L}$ with commutation relations
\begin{equation*}
    [t_a,t_b]=f_{ab}^ct_c\,.
\end{equation*}
Then the homological vector field on $\Pi \mathcal{L}$ reads
\begin{equation}\label{PL}
    Q=\frac12 c^bc^af_{ab}^d\frac{\partial}{\partial c^d}\,.
\end{equation}
By the definition of the parity reversing functor,
$\epsilon(c^a)=\epsilon(t_a)+1$. The linear space of functions on
$\Pi \mathcal{L}$ endowed with the differential $Q$ gives us a
model for the Chevaley-Eilenberg complex  of the Lie algebra
$\mathcal{L}$. Upon choosing a flat affine connection   on $\Pi
\mathcal{L}$, we see that the characteristic classes of $A$-series
are nothing but primitive elements of the Lie algebra cohomology:
\begin{equation*}
    {A}_{2n-1}=\mathrm{tr}(\mathrm{ad}_{a_1}\cdots
    \mathrm{ad}_{a_{2n-1}}) c^{a_1}\cdots c^{a_{2n-1}}\qquad
    \forall n\in \mathbb{N}\,.
\end{equation*}
Here $\mathrm{ad}_a=\{f_{ab}^c\}$ are the matrices of the adjoint
representation of $\mathcal{L}$.

The universal cocycles of $B$- and $C$-series are given by the
following ad-invariant tensors on $\Pi\mathcal{L}$:
\begin{equation*}
\begin{array}{l}
   \displaystyle  B_n=(\mathrm{ad}_{a_1}\cdots \mathrm{ad}_{a_n})^b_{a_{n+1}}
    dc^{a_1}\otimes \cdots \otimes dc^{a_{n+1}}\otimes \frac{\partial}{\partial
    c^b}\,,\\[5mm]
    C_n=\mathrm{tr}(\mathrm{ad}_{a_1}\cdots
    \mathrm{ad}_{a_n})dc^{a_1}\otimes \cdots\otimes dc^{a_n}\qquad
    \forall n\in \mathbb{N}\,.
\end{array}
\end{equation*}
Since any coboundary of (\ref{PL}) is necessarily proportional to
$c^a$, the tensor cocycles above are either zero or nontrivial.
For instance, if $\mathcal{L}$ is semi-simple, then the one-form
$C_1$ is zero, while the two-form $C_2$ is non-degenerate (the
Killing metric).

\end{document}